\documentclass[reqno,a4paper,final,12pt]{amsart}
\usepackage[british]{babel}
\usepackage{amssymb,amstext,amsthm,eucal,bbm,mathrsfs}
\usepackage{bbold}
\usepackage[margin=1.2in]{geometry}
\usepackage{array,color}
\usepackage[utf8]{inputenc}
\usepackage[notref,notcite]{showkeys}
\usepackage[small]{eulervm}
\usepackage{tgpagella}
\usepackage{hyperref}
\hypersetup{%
  pdftitle   = {The homogeneity theorem for supergravity backgrounds
    II: the six-dimensional theories},
  pdfkeywords = {supergravity, Killing spinors, isometries,
    homogeneity, supersymmetry},
  pdfauthor  = {José Figueroa-O'Farrill, Noel Hustler},
  pdfcreator = {\LaTeX\ with package \flqq hyperref\frqq}
}
\newcommand{\half}{\tfrac{1}{2}}

\newcommand{\fg}{\mathfrak{g}}

\newcommand{\fk}{\mathfrak{k}}

\newcommand{\fM}{\mathfrak{M}}
\newcommand{\fS}{\mathfrak{S}}

\newcommand{\fso}{\mathfrak{so}}

\newcommand{\Cl}{\mathrm{C}\ell}
\newcommand{\Spin}{\mathrm{Spin}}
\ifx\Sp\UnDeFiNeD
\newcommand{\Sp}{\mathrm{Sp}}
\else
\renewcommand{\Sp}{\mathrm{Sp}}
\fi
\newcommand{\USp}{\mathrm{USp}}

\newcommand{\SL}{\mathrm{SL}}

\newcommand{\RR}{\mathbb{R}}
\newcommand{\CC}{\mathbb{C}}
\newcommand{\HH}{\mathbb{H}}

\newcommand{\ZZ}{\mathbb{Z}}

\newcommand{\eD}{\mathcal{D}}

\newcommand{\eS}{\mathcal{S}}

\newcommand{\tV}{\mathbb{V}}

\newcommand{\rC}{\mathcal{C}}
\newcommand{\rL}{\mathscr{L}}
\newcommand{\rS}{\mathscr{S}}

\DeclareMathOperator{\End}{End}

\DeclareMathOperator{\im}{im}

\newcommand{\1}{\mathbbm{1}}

\theoremstyle{plain}
\newtheorem{lemma}{Lemma}
\newtheorem{proposition}[lemma]{Proposition}

\definecolor{orange}{rgb}{0.9,0.45,0}
\newcommand{\MUNCH}[1]{\relax}

\allowdisplaybreaks[1]
\begin{document}
\title[Homogeneity of supergravity backgrounds II]{The homogeneity
  theorem for supergravity backgrounds II: the six-dimensional theories}
\author[Figueroa-O'Farrill]{José Figueroa-O'Farrill}
\author[Hustler]{Noel Hustler}
\address{Maxwell and Tait Institutes, School of Mathematics, University of Edinburgh}
\thanks{EMPG-13-25}
\begin{abstract}
  We prove that supersymmetry backgrounds of (1,0) and (2,0)
  six-dimensional supergravity theories preserving more than one half
  of the supersymmetry are locally homogeneous.  As a byproduct we
  also establish that the Killing spinors of such a background
  generate a Lie superalgebra.
\end{abstract}
\maketitle
\tableofcontents

\section{Introduction}

The study of supergravity backgrounds, which had seen its first period
of intense activity in the 1980s in the context of Kaluza--Klein
supergravity (see \cite{DNP} for a then timely review), was retaken in
earnest in the mid-to-late 1990s ushered in by the
dualities-and-branes paradigm in string theory and by the
gauge/gravity correspondence.  Although as a result of this continuing
second effort a huge number of backgrounds are now known, it is fair
to say that we know very little about all but a few small corners of
the landscape of supergravity backgrounds.  This is perhaps not
surprising given our still incomplete knowledge about solutions to the
much more venerable four-dimensional Einstein--Maxwell equations.

The emphasis during the 1980s was on Freund--Rubin backgrounds, in
which the geometry (if not necessarily the fluxes) decompose into the
metric product of a four-dimensional spacetime and some internal
(typically compact) manifold, but the interest nowadays has widened to
backgrounds with an intricate and truly higher-dimensional geometry.
One particularly interesting class of backgrounds, due to the crucial
rôle they play in the gauge/gravity correspondence, comprises those
backgrounds preserving a substantial amount of the supersymmetry of
the theory.  In fact, the fraction of supersymmetry preserved by a
background has proved to be a very useful organising principle in our
efforts to tame the zoo of supergravity backgrounds, despite being a
rather coarse invariant.  Finer invariants, such as the holonomy group
of the gravitino connection, are much harder to calculate and hence
have yet to play a decisive role in the various classification efforts
underway.

One particularly attractive classification problem is that of
backgrounds which preserve a large fraction of the supersymmetry.  In
higher dimensions it is possible to make some limited progress by
working one's way down from the top; i.e., classifying maximally or
near-maximally supersymmetric backgrounds, but in order to make real
progress in that classification, new ideas seem to be required.  Based
on the increasing number of known backgrounds, Patrick Meessen (in a
private communication to the senior author in 2004) observed that
$>{}\half$-BPS backgrounds---i.e., those preserving more than half of
the supersymmetry---were homogeneous; that is, that the Lie group of
flux-preserving isometries of such a background acts transitively on
the underlying manifold.  He conjectured that this was always the case
and after some initial partial results \cite{FMPHom,EHJGMHom}, a local
version of the conjecture was recently demonstrated for ten- and
eleven-dimensional supergravity theories in
\cite{FigueroaO'Farrill:2012fp}.  In principle this ``reduces'' the
classification problem of $>{}\half$-BPS backgrounds to those
backgrounds which are homogeneous.  Alas, this is still a daunting
task; although progress can be made.

The purpose of this paper is to extend the homogeneity theorem in
\cite{FigueroaO'Farrill:2012fp} to two six-dimensional supergravity
theories: $(1,0)$ and $(2,0)$.  Other possible supergravity theories
in six dimensions are either not yet constructed or can be obtained by
dimensional reduction (without truncation) from higher-dimensional
theories, in which case the homogeneity theorem follows by general
arguments, as will be explained in a forthcoming paper.

For the $(1,0)$ theory, it was shown in \cite{Gutowski:2003rg} that
backgrounds preserve either all, half or none of the supersymmetry.
Therefore if a background preserves more than half of the
supersymmetry, it must be maximally supersymmetric and they are known
to be homogeneous.  Indeed, as shown in \cite{CFOSchiral}, such
backgrounds are lorentzian Lie groups with bi-invariant metrics and
invariant 3-form.  In this paper we give a different proof of this
result which has the virtue of not requiring the classification (hence
we prove a ``Theorem'' instead of a ``theorem'', in the
nomenclature of Victor Kac) and which results in addition in the
construction of the Killing superalgebra of the background.  We are
not aware of similar results for the $(2,0)$ theory beyond the
classification of maximally supersymmetric backgrounds in
\cite{CFOSchiral}.

The proof of the homogeneity theorem in
\cite{FigueroaO'Farrill:2012fp} consists of two steps. The first step
is to show that the natural squaring map from spinor fields to vector
fields, when applied to the Killing spinors of the background, yields
Killing vectors which also preserve the fluxes.  With a little extra
effort, and because it is an interesting result in its own right, one
also shows that the Killing spinors generate a Lie superalgebra,
called the Killing superalgebra of the background---a more refined
invariant than the fraction of supersymmetry, which only measures the
dimension of the odd subspace.

The second step is purely algebraic and consists in proving the
surjectivity of the squaring map restricted to any subspace of
dimension greater than one half the rank of the relevant spinor
bundle.  This uses the fact that the vector obtained by squaring a
spinor is causal.  Applied to the space of Killing spinors of
$>\half$-BPS backgrounds, it guarantees that the tangent space at
every point can be spanned by vector fields in the image of the
squaring map, which by the first step are known to be infinitesimal
symmetries of the background.  This proves the \emph{local}
homogeneity of the background.

The paper is organised as follows.  In
Section~\ref{sec:six-dimensional-1} we treat the $(1,0)$ theory, with
eight real supercharges.  We introduce the relevant notion of Killing
spinor and show that they generate a Lie superalgebra.  We then show
that if the space of Killing spinors has dimension greater than four,
then the background is locally homogeneous.  In
Section~\ref{sec:six-dimensional-2} we do the same for the $(2,0)$
theory, with sixteen real supercharges. The calculations here are more
complicated due to the presence of R-symmetry generators in the
definition of Killing spinors.  Again we find that the Killing spinors
generate a Lie superalgebra and when its odd subspace has dimension
greater than eight, the background is locally homogeneous.  Finally,
in Appendix~\ref{sec:summ-spin-results}, we collect the basic facts
about spinors in six dimensions which are used in the bulk of the
paper: pinor, spinor and R-symmetry representations, the invariant
inner products, explicit matrix realisations and some useful
consequences of the Clifford relations, including the relevant Fierz
identities.

\section{Six-dimensional (1,0) supergravity}
\label{sec:six-dimensional-1}

Let $(M,g,H)$ be a bosonic background of six-dimensional $(1,0)$
supergravity \cite{NishinoSezgin10}.  This means that $(M,g)$ is a
connected six-dimensional lorentzian spin manifold and $H \in
\Omega_-^3(M)$ a closed anti-selfdual three-form and they satisfy the
field equations of the theory with fermions equal to zero.  We shall
not need the equations in the following.

Let $S_+$ denote the positive-chirality spinor representation of
$\Spin(5,1)$.  It is a two-dimensional quaternionic representation,
but we prefer to work with complex representations, whence we will
think of $S_+$ as a four-dimensional complex representation with an
invariant quaternionic structure; that is, with a complex antilinear
map $J: S_+ \to S_+$ which obeys $J^2 = - 1$.

Similarly, the fundamental representation $S_1$ of the R-symmetry
group of $(1,0)$ supergravity, which is isomorphic to
$\Sp(1)$, is a two-dimensional complex representation with an
invariant quaternionic structure $j: S_1 \to S_1$.

The tensor product $S_+ \otimes_\CC S_1$ of these two representations
is an 8-dimensional complex representation with an invariant
conjugation given by $J \otimes j$, whence it is a complex
representation of real type.  In other words, it is the
complexification of a real representation $\eS_+$, defined by
\begin{equation}
  S_+ \otimes_\CC S_1 \cong \eS_+ \otimes_\RR \CC~.
\end{equation}
The real representation $\eS_+$, which is the real subspace of $S_+
\otimes_\CC S_1$ fixed under the conjugation, is eight-dimensional and
is the relevant spinorial representation for this supergravity theory.
With some abuse of language we will also denote by $\eS_+$ the spinor
bundle on $M$ associated to this representation.

Let $\eS_-$ be the real eight-dimensional representation defined
as $\eS_+$ but starting from the negative-chirality spinor
representation $S_-$ of $\Spin(5,1)$.  As shown in
Section~\ref{sec:underly-real-spin} in the Appendix, there is a
$(\Spin(5,1)\times \Sp(1))$-invariant symplectic structure on $\eS =
\eS_+ \oplus \eS_-$, denoted by $\left<-,-\right>$, and satisfying
\begin{equation}
  \left<\Gamma_a \varepsilon_1, \varepsilon_2\right> = -
  \left<\varepsilon_1, \Gamma_a \varepsilon_2\right>~.
\end{equation}

\subsection{The Killing superalgebra}
\label{sec:killing-superalgebra}

The supersymmetry variation of the gravitino defines a connection $D$
on the bundle $\eS_+$, defined for a spinor field $\varepsilon \in
C^\infty(M;\eS_+)$ and a vector field $X \in C^\infty(M;TM)$ by
\begin{equation}
  D_X \varepsilon = \nabla_X \varepsilon + \tfrac14 \iota_X H
  \cdot \varepsilon~,
\end{equation}
where $\nabla$ is the spin connection on $\eS_+$ induced by the
Levi-Civita connection, $\cdot$ is the Clifford action and $\iota_X$
denotes the contraction by the vector field $X$.  Spinor fields which
are parallel with respect to $D$ are called \emph{Killing spinors}.
They form a real vector space $\fg_1$ whose dimension is at most the
rank of $\eS_+$, since a parallel section of a bundle over a connected
manifold is uniquely determined by its value at any one point.  For
the theory in question, $\dim\fg_1 \leq 8$.  Once fixing a point $p
\in M$, we will freely identify $\fg_1$ with a subspace of the fibre
of $\eS_+$ at $p$ which we will in turn identify with the
representation $\eS_+$ itself.  In other words, we will often think of
$\fg_1$ as a subspace of $\eS_+$.

Let $\fg_0$ denote the Lie algebra of Killing vector fields on $(M,g)$
which preserve $H$.  We will show that on $\fg = \fg_0 \oplus \fg_1$
we can define the structure of a Lie superalgebra.  This is by now a
standard construction for supergravity theories
\cite{JMFKilling,FMPHom,EHJGMHom}.

The Lie superalgebra structure on $\fg = \fg_0 \oplus \fg_1$ is a
graded skew bilinear map $\fg \times \fg \to \fg$, which unpacks into
three bilinear maps:
\begin{enumerate}
\item a skewsymmetric bilinear map $[-,-]: \fg_0 \times \fg_0 \to
  \fg_0$, which is simply the Lie bracket of vector fields;
\item the action of $\fg_0$ on $\fg_1$, which is a bilinear map $[-,-]:\fg_0
  \times \fg_1 \to \fg_1$, defined by the spinorial Lie derivative
  (see, e.g., \cite{Kosmann})
  \begin{equation}
    [X,\varepsilon] := \rL_X \varepsilon = \nabla_X \varepsilon - \rho(\nabla
    X) \varepsilon
  \end{equation}
  where $\nabla X$ is the skewsymmetric endomorphism of $TM$ defined
  by $Y \mapsto \nabla_Y X$ and $\rho: \fso(TM) \to \End(\eS_+)$ is
  the spin representation; and
\item a symmetric bilinear map $[-,-]: \fg_1 \times fg_1 \to \fg_0$,
  whose restriction to the diagonal is called the \emph{squaring map}:
  it is essentially the transpose of the Clifford action of vectors on
  spinors and will be defined presently.
\end{enumerate}
These maps are then subject to the Jacobi identity, which unpacks into
four components.  As we will review presently, three of these
components are automatically zero, but the fourth needs proof, which
we provide.

The transpose of the Clifford action of $TM$ on $\eS$ under the metric
$g$ on $TM$ and the symplectic structure on $\eS$ defines a symmetric
bilinear map $\eS_+ \times \eS_+ \to TM$.  Being symmetric it is
uniquely determined by its restriction to the diagonal, a quadratic
map $\eS_+ \to TM$, sending a spinor field $\varepsilon$ to its
\emph{Dirac current} $V_\varepsilon$, defined by
\begin{equation}
  g(V_\varepsilon, X) = \left<\varepsilon,X\cdot \varepsilon\right>~.
\end{equation}
By the usual polarisation identity, we define the vector field
$[\varepsilon_1,\varepsilon_2]$ corresponding to two spinor fields
$\varepsilon_1,\varepsilon_2 \in C^\infty(M;\eS_+)$ by
\begin{equation}
  \label{eq:polarisation}
  2 [\varepsilon_1,\varepsilon_2] = V_{\varepsilon_1 + \varepsilon_2}
  - V_{\varepsilon_1} - V_{\varepsilon_2}~.
\end{equation}

The first result is that when $\varepsilon_1,\varepsilon_2$ are
Killing spinors, $[\varepsilon_1,\varepsilon_2]$ is a Killing vector
which preserves the 3-form $H$.  Clearly, it is enough to show this
for the Dirac current of a Killing spinor $\varepsilon$.  To see this,
we first observe that a Killing spinor $\varepsilon$ is parallel
relative to a connection $D$ defined by
\begin{equation}
  D_\mu = \nabla_\mu - \tfrac12 \rho(H_\mu)~,
\end{equation}
where $\rho$ is the spin representation applied to the skewsymmetric
endomorphism $H_\mu$ of $TM$ defined, relative to a pseudo-orthonormal
basis $e_a$, by
\begin{equation}
  H_\mu (e_a)  = H_\mu{}^b{}_a e_b~.
\end{equation}
In other words, $D$ is the spin connection corresponding to an affine
connection also denoted $D$ and defined by $D_\mu = \nabla_\mu -
\tfrac12 H_\mu$.  By covariance, if $D_\mu \varepsilon = 0$, then
$D_\mu V_\varepsilon = 0$ as well.  In other words, writing $V$ for
$V_\varepsilon$,
\begin{equation}
  D_\mu V_\nu = \nabla_\mu V_\nu - \tfrac12 H_{\mu\nu\rho} V^\rho =
  0~,
\end{equation}
whence $\nabla_\mu V_\nu = \tfrac12 V^\rho H_{\rho\mu\nu}$.  First of
all, we see that $\nabla_\mu V_\nu = - \nabla_\nu V_\mu$, whence $V$
is a Killing vector field.  We also see that
\begin{equation}
  d V^\flat = \tfrac12 \iota_V H
\end{equation}
whence $d\iota_V H = 0$.  Since $dH = 0$, this shows that $\rL_V H =
0$, whence $V$ preserves $H$.  Then after polarisation we obtain a
symmetric bilinear map $[-,-]: \fg_1  \times \fg_1 \to \fg_0$.

If $K \in \fg_0$ is any Killing vector field which preserves $H$, then
the Lie derivative $\rL_K$ leaves invariant the connection $D$:
\begin{equation}
  \rL_K D_X - D_X \rL_K = D_{[K,X]}~.
\end{equation}
In turn, this means that $\rL_K$ acting on spinors also leaves
invariant the spin connection $D$, whence it sends Killing spinors to
Killing spinors, defining a map $[-,-]= \fg_0 \times \fg_1 \to \fg_1$
by $[K,\varepsilon] = \rL_K \varepsilon$.

It now remains to prove the Jacobi identity for the bracket $[-,-]:
\fg \times \fg \to \fg$ just defined on $\fg = \fg_0 \oplus \fg_1$.
The only component of the Jacobi identity which needs to be checked is
the $(\fg_1,\fg_1,\fg_1)$-component.  This is given by a symmetric
trilinear map $\fg_1 \times \fg_1 \times \fg_1 \to \fg_1$ which again
is determined uniquely via polarisation by the restriction to the
diagonal: the map sending a Killing spinor $\varepsilon$ to the
Killing spinor $\rL_{V_\varepsilon} \varepsilon$.  We need to show
that this is zero for all $\varepsilon \in \fg_1$.

A quick calculation shows that
\begin{equation}
  \rL_{V_\varepsilon} \varepsilon = \iota_{V_\varepsilon}H \cdot
  \varepsilon~,
\end{equation}
whose vanishing, using the first equation~\eqref{eq:cliffordintext}, becomes
\begin{equation}
  V_\varepsilon \cdot H \cdot \varepsilon + H \cdot V_\varepsilon
  \cdot \varepsilon = 0~.
\end{equation}

From Lemma~\ref{lem:asdkillspos}, we know that $H\cdot \varepsilon =
0$, whence the first term vanishes.  The second term will also vanish
as a consequence of the following

\begin{proposition}
  \label{prop:Vkillsepsilon}
  Let $\varepsilon \in \eS_+$ and $V_\varepsilon$ its Dirac current.
  Then $V_\varepsilon \cdot \varepsilon = 0$.
\end{proposition}

\begin{proof}
  We have
  \begin{equation}
   V_\varepsilon \cdot \varepsilon  = \left<\varepsilon,\Gamma^a
      \varepsilon\right> \Gamma_a \varepsilon = \epsilon_{AB}
    \left(\varepsilon^A,\Gamma^a\varepsilon^B\right)\Gamma_a
    \varepsilon~,
  \end{equation}
  where we have expanded $\varepsilon$ in terms of a symplectic basis
  for the fundamental representation $S_1$ of the R-symmetry group
  $\USp(2)$, as described in Section~\ref{sec:1-0-fierz} of the
  appendix.  But then by Lemma~\ref{lem:3form10}, it vanishes.
\end{proof}

In summary, on $\fg = \fg_0 \oplus \fg_1$ we have the structure of a
Lie superalgebra, called the \emph{symmetry superalgebra} of the
supersymmetric $(1,0)$ background $(M,g,H)$.  The ideal $\fk =
[\fg_1,\fg_1] \oplus \fg_1$ generated by $\fg_1$ is called the
\emph{Killing superalgebra} of the background.

\subsection{Homogeneity}
\label{sec:homogeneity}

We will now prove the strong version of the (local) homogeneity
conjecture: that the even part of the Killing superalgebra already
acts locally transitively on the background.

It follows from Proposition~\ref{prop:Vkillsepsilon} that the Dirac
current $V_\varepsilon$ of a chiral spinor $\varepsilon \in \eS_+$ is
null:
\begin{equation}
  g(V_\varepsilon,V_\varepsilon) = \left<\varepsilon, V_\varepsilon
    \cdot \varepsilon\right> = 0~.
\end{equation}

The proof of homogeneity follows the same steps in
\cite{FigueroaO'Farrill:2012fp}, which we briefly review for the sake
of completeness.

Let $\dim\fg_1 > 4 = \tfrac12 \dim\eS_+$.  We want to show that for
each $p \in M$, the symmetric bilinear map $\varphi: \fg_1 \times
\fg_1 \to T_pM$, obtained by sending the pair
$(\varepsilon_1,\varepsilon_2)$ of Killing spinors to the tangent
vector $[\varepsilon_1,\varepsilon_2](p)$ to $M$ at $p$ is surjective.
Let $v \in T_pM$ be perpendicular to the image of $\varphi$; that is,
to $[\varepsilon_1,\varepsilon_2](p)$ for all
$\varepsilon_{1,2}\in\fg_1$.  This means that for all
$\varepsilon_{1,2} \in \fg_1$,
\begin{equation}
  \left<\varepsilon_1, v \cdot \varepsilon_2\right>= 0~,
\end{equation}
or that Clifford product by $v$ maps $\fg_1$ to $\fg_1^\perp \subset
\eS_-$.  Since $\dim\fg_1 > \dim\fg_1^\perp$, it follows that the
Clifford product by $v$ has nontrivial kernel and hence that $v$ is
null, since by the Clifford relation $v^2 \cdot \varepsilon = -g(v,v)
\varepsilon$.  Every vector which is perpendicular to the image of
$\varphi$ is therefore null and hence $(\im\varphi)^\perp \subset
T_pM$ is an isotropic subspace.  Since the isotropic subspaces of
$T_pM$ are at most one-dimensional, we have two possibilities: either
$\varphi$ is surjective or else $(\im\varphi)^\perp$ is
one-dimensional and spanned by a null vector $n$, say.  In this latter
case, the Dirac current $V_\varepsilon$ of every Killing spinor
$\varepsilon \in \fg_1$ is a null vector perpendicular to $n$, whence
it has to be proportional to $n$, otherwise they would span a
two-dimensional isotropic subspace.  But then by polarisation, every
vector in the image of $\varphi$ would be proportional to $n$,
contradicting the fact that $\im\varphi$ has codimension one.

In summary, we have shown that the tangent space to $M$ at any point
$p$ is spanned by the values at $p$ of Killing vectors in
$[\fg_1,\fg_1]$.  This shows that the $(1,0)$ background $(M,g,H)$ is
locally homogeneous.

\section{Six-dimensional (2,0) supergravity}
\label{sec:six-dimensional-2}

A bosonic background $(M,g,H)$ of $(2,0)$ supergravity
\cite{Townsend20,Tanii20} consists of a connected, lorentzian, spin
6-dimensional manifold $(M,g)$ and a closed anti-selfdual $\tV$-valued
three-form $H \in \Omega^3_-(M;\tV)$, where $\tV$ is the real
5-dimensional orthogonal representation of the $\USp(4) \cong
\Spin(5)$ R-symmetry group of the theory.  We may choose an
orthonormal basis $e_i$ for $\tV$ and hence think of $H = H^i e_i$ as
five anti-selfdual three-forms $H^i \in \Omega_-^3(M)$.  As in the
$(1,0)$ theory, these fields are subject to the field equations of the
theory with fermions put to zero, but we shall not need their explicit
form in what follows.

As before, $S_\pm$ are the complex 8-dimensional irreducible spinor
representations of $\Spin(5,1)$ and now $S_2$ denotes the fundamental
representation of $\USp(4)$, which is complex and 4-dimensional.  Both
$S_\pm$ and $S_2$ have invariant quaternionic structures, whence their
tensor product is a complex representation of $\Spin(5,1) \times
\USp(4)$ of real type, whence the complexification of a
sixteen-dimensional real representation $\rS_\pm$.  We will let $\rS =
\rS_+ \oplus \rS_-$, on which we have an action of $\Cl(5,1) \otimes
\Cl(0,5)$ with generators $\Gamma_a$ for $\Cl(5,1)$ and $\gamma_i$ for
$\Cl(0,5)$.  As discussed in the Appendix, $\rS$ has a symplectic
inner product $\left<-,-\right>$ relative to which $\rS_\pm$ are
lagrangian subspaces and such that
\begin{equation}
  \left<\varepsilon_1, \Gamma_a \varepsilon_2\right> = -
  \left<\Gamma_a \varepsilon_1, \varepsilon_2\right>
  \qquad\text{and}\qquad
  \left<\varepsilon_1, \gamma_i \varepsilon_2\right> = +
  \left<\gamma_i \varepsilon_1, \varepsilon_2\right>~.
\end{equation}

\subsection{The Killing superalgebra}
\label{sec:killing-superalgebra-20}

A \emph{Killing spinor} of $(2,0)$ supergravity is a section
$\varepsilon$ of $\rS_+$ which is parallel relative to a connection
$\eD$ defined by
\begin{equation}
  \eD_\mu \varepsilon = \nabla_\mu \varepsilon + \tfrac18 H^i_{\mu ab}
  \Gamma^{ab}\gamma_i \varepsilon~.
\end{equation}

The \emph{Dirac current} $V_\varepsilon$ of a spinor $\varepsilon\in
C^\infty(M;\rS_+)$ is the vector field defined by
\begin{equation}
  g(V_\varepsilon,X) = \left<\varepsilon, X \cdot \varepsilon\right>~,
\end{equation}
for all vector fields $X$.  Its coefficients relative to an
orthonormal frame are then given by $V_\varepsilon^a =
\left<\varepsilon,\Gamma^a\varepsilon\right>$.

As before, the Dirac current of a Killing spinor is a Killing vector
which preserves $H$.  Indeed, let $\varepsilon$ be a Killing spinor
and let $V = V_\varepsilon$ denote its Dirac current.  Its covariant
derivative is given by
\begin{equation}
  \begin{split}
    \nabla_\mu V^\nu &= \nabla_\mu \left<\varepsilon, \Gamma^\nu
      \varepsilon\right>\\
    &= \left<\nabla_\mu \varepsilon, \Gamma^\nu \varepsilon\right> + \left<\varepsilon, \Gamma^\nu
      \nabla_\mu \varepsilon\right>\\
    &= -\tfrac18 H^i_{\mu\rho\sigma} \left(
      \left<\Gamma^{\rho\sigma}\gamma_i \varepsilon, \Gamma^\nu
        \varepsilon\right> + \left<\varepsilon,\Gamma^\nu
        \Gamma^{\rho\sigma} \gamma_i \right>\right)\\
    &= \tfrac18 H^i_{\mu\rho\sigma} \left<\gamma_i \varepsilon, [\Gamma^{\rho\sigma},\Gamma^\nu]
      \varepsilon\right>\\
    &= \tfrac12 H^i_\mu{}^\nu{}_\sigma \left<\gamma_i \varepsilon,
      \Gamma^\sigma \varepsilon\right>~.
  \end{split}
\end{equation}
This can be rewritten as
\begin{equation}
  \nabla_\mu V_\nu = \tfrac12 H^i_{\mu\nu\rho}
  \left<\varepsilon,\Gamma^\rho\gamma_i \varepsilon\right>~,
\end{equation}
which shows that $V$ is a Killing vector.

Let $\theta \in \Omega^1(M;\tV)$ be the $\tV$-valued one-form defined by
$\theta^i_\mu = \left<\varepsilon,\Gamma_\mu\gamma^i
  \varepsilon\right>$.  Its covariant derivative is given by
\begin{equation}
  \begin{split}
    \nabla_\mu \theta^i_\nu &=
    \left<\nabla_\mu\varepsilon,\Gamma_\nu\gamma^i \varepsilon\right>
    + \left<\varepsilon,\Gamma_\nu\gamma^i
      \nabla_\mu\varepsilon\right>\\
    &= -\tfrac18 H^j_{\mu\rho\sigma} \left(
      \left<\Gamma^{\rho\sigma}\gamma_j \varepsilon, \Gamma_\nu
        \gamma^i \varepsilon\right> + \left<\varepsilon, \Gamma_\nu
        \Gamma^{\rho\sigma} \gamma^i \gamma_j
        \varepsilon\right>\right)\\
    &= \tfrac18 H^j_{\mu\rho\sigma} \left(
      \left<\gamma^i \gamma_j \varepsilon, \Gamma^{\rho\sigma}\Gamma_\nu
        \varepsilon\right> - \left< \gamma_j \gamma^i \varepsilon, \Gamma_\nu
        \Gamma^{\rho\sigma} \varepsilon\right>\right)~.
  \end{split}
\end{equation}
Using the Clifford relations $\gamma_j \gamma^i = \delta_j^i +
\gamma_j{}^i$, we can rewrite this as
\begin{equation}
  \begin{split}
    \nabla_\mu \theta^i_\nu  &= \tfrac18 H^i_{\mu\rho\sigma}
    \left<\varepsilon, [\Gamma^{\rho\sigma},\Gamma_\nu]
      \varepsilon\right> - \tfrac18 H^j_{\mu\rho\sigma}
    \left<\gamma_j{}^i \varepsilon, (\Gamma^{\rho\sigma}\Gamma_\nu +
      \Gamma_\nu \Gamma^{\rho\sigma}) \varepsilon\right>\\
    &= \tfrac12 H^i_{\mu\nu\rho}
    \left<\varepsilon,\Gamma^\rho\varepsilon\right> + \tfrac18
    H^j_\mu{}^{\rho\sigma} \left<\gamma^i{}_j \varepsilon,
      (\Gamma_{\rho\sigma}\Gamma_\nu + \Gamma_\nu\Gamma_{\rho\sigma})
      \varepsilon \right>\\
    &= \tfrac12 H^i_{\mu\nu\rho}
    \left<\varepsilon,\Gamma^\rho\varepsilon\right> + \tfrac14
    H^j_\mu{}^{\rho\sigma} \left<\gamma^i{}_j \varepsilon, \Gamma_{\nu\rho\sigma}\varepsilon\right>~.
  \end{split}
\end{equation}

It follows from this that its exterior derivative $d\theta \in
\Omega^2(M,\tV)$, with components $(d\theta)^i_{\mu\nu} = \nabla_\mu
\theta^i_\nu - \nabla_\nu \theta^i_\mu$, is given by
\begin{equation}
 (d\theta)^i_{\mu\nu} = H^i_{\mu\nu\rho}
  \left<\varepsilon,\Gamma^\rho \varepsilon\right>+ \tfrac14 H^j_\mu{}^{\rho\sigma}
  \left<\gamma^i{}_j\varepsilon, \Gamma_{\nu\rho\sigma} \varepsilon
  \right> - \tfrac14 H^j_\nu{}^{\rho\sigma}
  \left<\gamma^i{}_j\varepsilon, \Gamma_{\mu\rho\sigma} \varepsilon
  \right>~.
\end{equation}
Notice that the last two terms can be written in terms of a Clifford
commutator, so that
\begin{equation}
  (d\theta)^i_{\mu\nu} = H^i_{\mu\nu\rho}
  \left<\varepsilon,\Gamma^\rho \varepsilon\right> + \tfrac1{24} H^j_{\rho\sigma\tau}
  \left<\gamma^i{}_j\varepsilon, [\Gamma_{\mu\nu},\Gamma^{\rho\sigma\tau}] \varepsilon
  \right>~.
\end{equation}
The second term in the RHS is seen to vanish, since $[\Gamma_{\mu\nu},
H^j] \varepsilon = 0$, because $H^j$ and hence also its infinitesimal
Lorentz transformation $[\Gamma_{\mu\nu}, H^j]$ are anti-selfdual and
hence annihilate $\varepsilon$ by Lemma~\ref{lem:asdkillspos} in the
appendix.  The remaining term in the RHS is precisely the contraction
of $H$ by $V$.  This shows that $\iota_V H = d\theta$ is closed and,
since so is $H$, that $\rL_V H = d\iota_V H + \iota_V dH = 0$, showing
that $V$ leaves $H$ invariant.

We therefore have all the ingredients for a Lie superalgebra on the
vector space $\fg = \fg_0 \oplus \fg_1$, where $\fg_0$ is the Lie
algebra of Killing vector fields which in addition preserve $H$ and
$\fg_1$ is the space of Killing spinors, which for the $(2,0)$ theory
has dimension at most $16$.  The bracket $[-,-]:\fg_0 \times \fg_1 \to
\fg_1$ is given by the spinorial Lie derivative $[K,\varepsilon] =
\rL_K \varepsilon$, which since $K \in \fg_0$ leaves $\eD$ invariant and
hence takes Killing spinors to Killing spinors.  The bracket $[-,-]:
\fg_1 \times \fg_1 \to \fg_0$ is obtained as before by polarising the
construction of the Dirac current, as in
equation~\eqref{eq:polarisation}.

Three of the four components of the Jacobi identity vanish by
construction, whence only the $(\fg_1,\fg_1,\fg_1)$ component needs to
be checked.  This is a symmetric trilinear map $\fg_1 \times \fg_1
\times \fg_1 \to \fg_1$, whence it vanishes if and only if it vanishes
when restricted to the diagonal, which is the map sending a Killing
spinor $\varepsilon$ to its Lie derivative along its Dirac current:
$\rL_{V_\varepsilon} \varepsilon$.

Letting $V = V_\varepsilon$, we have
\begin{equation}
  \begin{split}
    \rL_V \varepsilon &= \nabla_V\varepsilon - \rho(\nabla V)
    \varepsilon\\
    &= V^\mu \nabla_\mu \varepsilon - \tfrac14 \nabla_\mu V_\nu
    \Gamma^{\mu\nu} \varepsilon\\
    &= -\tfrac18 H^i_{\mu\nu\rho}\left( \Gamma^{\mu\nu} \theta_i^\rho +
      V^\mu \Gamma^{\nu\rho}\gamma_i \right) \varepsilon\\
    &= -\tfrac18 H^i_{\mu\nu\rho}\Gamma^{\mu\nu} \left( \theta_i^\rho +
      V^\rho \gamma_i \right) \varepsilon\\
    &= \tfrac1{48} H^i_{\mu\nu\rho}\left(\Gamma^{\mu\nu\rho}
      \Gamma_\sigma + \Gamma_\sigma \Gamma^{\mu\nu\rho}\right) \left(
      V^\sigma \gamma_i + \theta_i^\sigma\right) \varepsilon~.
  \end{split}
\end{equation}

We now use that $H^i$ Clifford annihilates $\varepsilon$
(Lemma~\ref{lem:asdkillspos} in the appendix) to arrive at
\begin{equation}
  \label{eq:111Jacobi20}
  \begin{split}
    \rL_V \varepsilon &= \tfrac1{48} H^i_{\mu\nu\rho} \Gamma^{\mu\nu\rho}
    \Gamma_\sigma \left( V^\sigma \gamma_i + \theta_i^\sigma\right)
    \varepsilon\\
    &= \tfrac1{48} H^i_{\mu\nu\rho}\Gamma^{\mu\nu\rho} \Gamma_\sigma \left(
      \left<\varepsilon,\Gamma^\sigma\varepsilon \right> \gamma_i +
      \left<\varepsilon,\Gamma^\sigma\gamma_i \varepsilon
      \right>\right) \varepsilon~.
  \end{split}
\end{equation}

This can be rewritten in a way that allows us to use the Fierz
identity~\eqref{eq:20Fierz}, namely
\begin{equation}
  \rL_V \varepsilon = \tfrac1{48} H^i_{\mu\nu\rho}\Gamma^{\mu\nu\rho} \Gamma_\sigma
  \left( \gamma_i (\varepsilon\otimes\varepsilon^\flat) +
    (\varepsilon\otimes\varepsilon^\flat) \gamma_i \right)
  \Gamma^\sigma\varepsilon~.
\end{equation}
Using that Fierz identity and also equation~\eqref{eq:cliffrels}, we
may rewrite this finally as
\begin{equation}
  \begin{split}
    \rL_V \varepsilon &= -\tfrac1{24}
    H^i_{\mu\nu\rho}\Gamma^{\mu\nu\rho} \Gamma_\sigma
    \left(\left<\varepsilon,\Gamma^\sigma\varepsilon\right> \gamma_i +
      \left<\varepsilon,\Gamma^\sigma\gamma_i\varepsilon\right>\right)
    \varepsilon\\
    &= -\tfrac1{24}
    H^i_{\mu\nu\rho}\Gamma^{\mu\nu\rho} \Gamma_\sigma
    \left(V^\sigma \gamma_i + \theta_i^\sigma\right) \varepsilon~.
  \end{split}
\end{equation}
Comparing with equation~\eqref{eq:111Jacobi20}, we see that it must
vanish.

This shows that the brackets thus defined on $\fg = \fg_0 \oplus
\fg_1$ turn it into a Lie superalgebra.  The ideal $\fk =
[\fg_1,\fg_1] \oplus \fg_1$ is the \emph{Killing superalgebra} of the
$(2,0)$ background $(M,g,H)$.

\subsection{Homogeneity}
\label{sec:homogeneity-20}

The proof of homogeneity follows the same steps as in the $(1,0)$
theory.  An essential ingredient is that the Dirac current of a spinor
is a causal vector; that is, either null or timelike.  In the $(1,0)$
theory we showed that the Dirac current is always null, but in the
$(2,0)$ theory this is not the case.  Nevertheless we can still show
that it is causal.  We have only managed to do this by an explicit
computation using the realisation in Section~\ref{sec:expl-matr-real}
in the appendix.

The Dirac current of $\varepsilon \in \rS_+$ has components
\begin{equation}
  K^a = \left<\varepsilon,\Gamma^a\varepsilon\right>~.
\end{equation}
For nonzero $\varepsilon \in \rS_+$, it follows that $\Gamma^a \varepsilon \in
\rS_-$ and hence we may express the inner product in either bilinear or
sesquilinear forms.  We choose the sesquilinear form and compute the
0th component of the Dirac current.  In the notation of
Section~\ref{sec:expl-matr-real}, we find
\begin{equation}
  \begin{split}
    K^0 &= \varepsilon^\dagger (B \otimes b) (\Gamma^0 \otimes \1_4) \varepsilon\\
    &= \varepsilon^\dagger (- \Gamma_{12345}\Gamma^0 \otimes \1_4)
    \varepsilon\\
    &= \varepsilon^\dagger (\Gamma_7 \otimes \1_4) \varepsilon\\
    &= \varepsilon^\dagger \varepsilon > 0~,
  \end{split}
\end{equation}
where we have used that $\Gamma_7 \varepsilon = \varepsilon$.  This
shows that $K^0$ never vanishes and thus $K$ cannot be spacelike,
otherwise we could Lorentz transform to a frame where $K^0=0$.

The proof now follows \emph{mutatis mutandis} the same steps as those
outlined in Section~\ref{sec:homogeneity} for the $(1,0)$ case and
will not be repeated here.  In summary, if the dimension of the space
of Killing spinors is greater than 8, then the $(2,0$) background
$(M,g,H)$ is locally homogeneous.

In summary, we have established the existence of Killing superalgebras
for supersymmetric backgrounds of six-dimensional $(1,0)$ and $(2,0)$
supergravities and used that to show that $>\half$-BPS backgrounds are
locally homogeneous.  Together with the results of
\cite{FigueroaO'Farrill:2012fp}, and using that the homogeneity
theorem survives dimensional reduction, this establishes the validity
of the homogeneity theorem for all (pure, Poincaré) supergravity
theories which have been constructed thus far.  Details will appear in
a forthcoming paper.

\section*{Acknowledgments}

This work was supported in part by the grant ST/J000329/1 ``Particle
Theory at the Tait Institute'' from the UK Science and Technology
Facilities Council.

\appendix

\section{Summary of spinorial results}
\label{sec:summ-spin-results}

\subsection{The Clifford module and its inner products}
\label{sec:clifford-module-its}

Our Clifford algebra conventions follow \cite{Harvey}.  We define
$\Cl(s,t)$ to be the Clifford algebra associated with the real vector
space $\RR^{s+t}$ with inner product given by the matrix
\begin{equation}
  \eta =  \begin{pmatrix} \1_s & 0 \\ 0 & -\1_t \end{pmatrix}~,
\end{equation}
where $\1_p$ is the $p\times p$ identity matrix.  This means that
$\Cl(s,t)$ is the associative unital algebra generated by $\Gamma_a$,
$a=1,\dots,s+t$, subject to the relations
\begin{equation}
  \Gamma_a \Gamma_b + \Gamma_b \Gamma_a = - 2 \eta_{ab} \1~.
\end{equation}
(Notice the sign!)  In this paper we are interested in $\Cl(5,1)$.

As a real associative algebra, $\Cl(5,1)$ is isomorphic to the algebra
$\HH(4)$ of $4\times 4$ quaternionic matrices.  This means that it has
a unique irreducible representation, $\fM$, which is quaternionic and of
dimension $4$.  We prefer, however, to work over the complex numbers,
so that we will represent $\Gamma_a$ as complex $8 \times 8$ matrices
leaving invariant a quaternionic structure.  The resulting
8-dimensional complex representation is the complex vector space $P$
obtained from the right quaternionic vector space $\fM$ with via
restriction of scalars to $\CC$.  We call $P$ the \emph{pinor
  representation} of $\Cl(5,1)$.

There are two natural involutions of $\Cl(5,1)$, each one realisable
as the adjoint relative to a quaternionic inner product on $\fM$.  The
two inner products are denoted $\left<-,-\right>_\pm$ and defined by
\begin{equation}
  \label{eq:spinor-ips}
  \left<\Gamma_a \varepsilon_1, \varepsilon_2\right>_\pm = \pm
  \left<\varepsilon_1, \Gamma_a \varepsilon_2\right>~,
\end{equation}
where $\left<-,-\right>_+$ is $\HH$-hermitian, and
$\left<-,-\right>_-$ is $\HH$-skewhermitian.  These quaternionic inner
products induce inner products on the pinor representation $P$.  This
is done by decomposing $\left<-,-\right>_\pm$, which are $\HH$-valued,
into $\CC$-valued inner products:
\begin{equation}
  \left<-,-\right>_+ = h_+(-,-) + j \omega_+(-,-)~,
\end{equation}
where $h_+$ is $\CC$-hermitian and $\omega_+$ is $\CC$-symplectic; and
\begin{equation}
  \left<-,-\right>_- = i h_-(-,-) + j  g_-(-,-)~,
\end{equation}
where $h_-$ is $\CC$-hermitian and $g_-$ is $\CC$-symmetric. In either
case, one determines the other: $\omega_+(\varepsilon_1,\varepsilon_2)
= h_+(\varepsilon_1 j, \varepsilon_2)$ and
$g_-(\varepsilon_1,\varepsilon_2) = ih_-(\varepsilon_1 j,
\varepsilon_2)$.

\subsection{The spinor representations}
\label{sec:spin-repr}

The spin group $\Spin(5,1) \cong \SL(2,\HH)$ is contained in
$\Cl(5,1)$ and hence the irreducible Clifford module $\fM$ decomposes
under $\Spin(5,1)$ into the direct sum of two irreducible spinor
modules $\fS_\pm$, labelled by their chirality, i.e., the eigenvalue
of the volume element $\Gamma_7 = \Gamma^{012345}$ in $\Cl(5,1)$,
which obeys $\Gamma_7^2 = \1$.  The volume element $\Gamma_7$ is not
in the centre of $\Cl(5,1)$, but it commutes with $\Spin(5,1)$, whence
its eigenspaces $\fS_\pm$ are preserved by $\Spin(5,1)$.  These are
the positive- and negative-chirality spinor representations of
$\Spin(5,1)$.  They are quaternionic and of dimension two, but we will
again restrict scalars to obtain four-dimensional complex
representations $S_\pm$ with an invariant quaternionic structure.
This means that under $\Spin(5,1)$, $P = S_+ \oplus S_-$.  There is no
$\Spin(5,1)$-invariant inner product on $S_\pm$ (or $\fS_\pm$), but of
course there is on their direct sum, relative to which $S_\pm$ are
isotropic subspaces.  This means that $S_- = S_+^*$.

\subsection{The R-symmetry representations}
\label{sec:r-symm-repr}

The R-symmetry group of the $d=6$ $(p,q)$ supersymmetry
algebra is $\USp(2p)\times \USp(2q)$, whence $\USp(2) \cong \Sp(1)$
for the $(1,0)$ theory and $\USp(4) \cong \Sp(2)$ for the $(2,0)$
theory.  The spinor parameters in the supergravity theory transform
according to the fundamental representations of these groups, which
are quaternionic representations $\fS_1 \cong \HH$ for the $(1,0)$
theory and $\fS_2 \cong \HH^2$ for the $(2,0)$ theory.  Restricting
scalars to $\CC$ we arrive at complex representations $S_1$, of
dimension two, and $S_2$ of dimension four, with invariant
quaternionic structures, respectively.

The representations $S_1$ and $S_2$ have $\CC$-hermitian inner
products invariant under $\USp(2)$ and $\USp(4)$, respectively.
However the gravitino connection in the $(2,0)$ theory uses explicitly
an equivariant bilinear map $\tV \times S_2 \to S_2$, where $\tV$ is the real
5-dimensional representation of $\USp(4) \cong \Spin(5)$.  There are
precisely two such maps, corresponding to the Clifford actions of
$\Cl(\tV)\cong \Cl(0,5)$ on either of its two irreducible Clifford
modules.  This means that $S_2$ is to be thought of not just as a
spinor representation of $\Spin(5)$, but actually as one of the two
pinor representations of $\Cl(0,5)$.

As a real associative algebra, $\Cl(0,5)$ is isomorphic to two copies
of the algebra $\HH(2)$ of $2\times 2$ quaternionic matrices.
Therefore it has two inequivalent irreducible representations, which
are quaternionic of dimension $2$ or, after restricting scalars,
complex of dimension $4$ with an invariant quaternionic structure.
Let us call these latter complex representations $S_2$ and $S'_2$.
The action of $\Cl(0,5)$ is via $4\times 4$ complex matrices
$\gamma_i$, satisfying $\gamma_i \gamma_j + \gamma_j \gamma_i =
2\delta_{ij} \1$.  The two representations are distinguished by the
action of the volume element $\gamma_6 = \gamma_{12345}$, which is
central in $\Cl(5)$, satisfies $\gamma_6^2 = 1$ and acts like the
identity on $S_2$.  The $\Spin(5)$-invariant inner product on $S_2$ is
such that
\begin{equation}
  \left<\gamma_i \varepsilon_1, \varepsilon_2\right>= +
  \left<\varepsilon_1, \gamma_i \varepsilon_2\right>~.
\end{equation}
Indeed, with the opposite sign the volume element would be
skewsymmetric making $S_2$ isotropic.

\subsection{The underlying real spinorial representations}
\label{sec:underly-real-spin}

In the six-dimensional supergravity theories, the spinor
parameters of the supersymmetry transformations take values in a real
representation whose complexification is the tensor product of the
chiral spinor representation of $\Spin(5,1)$ and the fundamental
representation of the R-symmetry group.  As discussed above, these
representations are complex of quaternionic type and hence their
tensor product (over $\CC$) is a complex representation of real type
and thus the complexification of a real representation.  In this section
of the appendix we provide the details.

For brevity, we will consider the more general case of a tensor
product $V \otimes_\CC W$ of two complex representations of
quaternionic type.  This means that $V$ and $W$ have invariant
quaternionic structures $J_V$ and $J_W$, respectively.  They are
complex \emph{anti}linear maps which square to $-\1$.  Their tensor
product $c = J_V \otimes J_W$ is a complex antilinear map squaring to
$\1$---i.e., a conjugation.  The eigenspace of $c$ with eigenvalue $1$
is a real subrepresentation $U$ of $V \otimes_\CC W$ and indeed $V
\otimes_\CC W = U \otimes_\RR \CC = U \oplus i U$.  The complex inner products
on $V$ and $W$ (induced from the quaternionic inner products on the
original quaternionic representations) determine real inner products
on $U$.

Let us apply this now to the cases of interest.  As we have seen, the
spinor representations $S_\pm$ are isotropic, so if we want an inner
product we must work with the pinor representation $P = S_+ \oplus
S_-$.  As we saw in Section~\ref{sec:clifford-module-its}, on $P$ we
have one of two possible pairs of inner products: one pair consisting
of a $\CC$-hermitian and a $\CC$-symplectic inner product, and another
pair consisting of a $\CC$-skewhermitian and a $\CC$-symmetric inner
product.  We will choose the latter, in order for the Lie bracket on
the Killing spinors to be symmetric and hence have a chance of
generating a Lie superalgebra with nontrivial odd subspace.  This
means that we choose the inner product $\left<-,-\right>_-$ on $P$.

On $P\otimes_\CC S_1$ we therefore have a $\CC$-symplectic structure,
consisting of the tensor product of the $\CC$-symmetric and
$\CC$-symplectic inner products on $P$ and $S_1$, respectively.  This
restricts to a real symplectic inner product on the underlying real
subrepresentation $\eS$ of $P \otimes_\CC S_1$.  We will denote it by
$\left<-,-\right>$ and simply notice that the subspaces $\eS_\pm$,
defined as the underlying real representations of $S_\pm \otimes_\CC
S_1$, are lagrangian subspaces.

For the $(2,0)$ theory, we again pick the $\CC$-symmetric inner
product on $P$ and the $\CC$-symplectic inner product on $S_2$, so
that on $P \otimes_\CC S_2$ we have a $\CC$-symplectic structure,
restricting to a real symplectic inner product on the underlying real
representation denoted $\rS$ of $P \otimes_\CC S_2$.  We will again
denote it by $\left<-,-\right>$ and again notice that $\rS = \rS_+
\oplus \rS_-$, where the lagrangian subspaces $\rS_\pm$ are now defined
as the underlying real representations of $S_\pm \otimes_\CC S_2$.

\subsection{Explicit matrix realisation}
\label{sec:expl-matr-real}

An essential ingredient in the proof of the homogeneity theorem is the
fact that the Dirac current of a spinor is a causal vector.  Whereas
for the $(1,0)$-theory, this fact admits a rather elegant proof, for
the $(2,0)$-theory we have only managed to show this by calculating
using an explicit matrix realisation.  For completeness, and because
it may be useful in the future, we record here the necessary
formulae.  We let $\1_n$ denote the $n\times n$ identity
matrix and $\sigma_i$ the (hermitian) Pauli spin matrices with
$\sigma_1\sigma_2 = i \sigma_3$, et cetera.  An explicit realisation
for the generators $\Gamma_a$ of $\Cl(5,1)$ is given by the following
matrices:
\begin{equation}
  \begin{aligned}[m]
    \Gamma_0 &= \1 \otimes \1 \otimes \sigma_3\\
    \Gamma_1 &= -i\1 \otimes \sigma_1 \otimes \sigma_1\\
    \Gamma_2 &= -i\1 \otimes \sigma_2 \otimes \sigma_1
  \end{aligned}
  \qquad\qquad
  \begin{aligned}[m]
    \Gamma_3 &= i\sigma_1 \otimes \sigma_3 \otimes \sigma_1\\
    \Gamma_4 &= i\sigma_3 \otimes \sigma_3 \otimes \sigma_1\\
    \Gamma_5 &= i\1 \otimes \1 \otimes \sigma_2~,
  \end{aligned}
\end{equation}
with $\1 = \1_2$.

The invariant quaternionic structure is given by the composition $J =
m_J \circ \chi$, where $\chi$ is complex conjugation and $m_J$ is a
matrix which obeys $\Gamma_a m_J = m_J \overline{\Gamma_a}$
(invariance) and in addition $m_J \overline{m_J} = -\1_8$.  Invariance
says that $m_J$ commutes with the real $\Gamma_a$, namely
$\Gamma_{0,2,5}$, and anticommutes with the imaginary $\Gamma_a$,
namely $\Gamma_{1,3,4}$.  Thus we can take $m_J = \Gamma_{025}$, which
is real and obeys $\Gamma_{025}^2 = -\1_8$.

The $\HH$-skewhermitian inner product $\left<-,-\right>_-$ decomposes
into $ih_- + jg_-$, where $ih_-$ is $\CC$-skewhermitian and $g_-$ is
$\CC$-symmetric.  In this explicit realisation, $ih_-$ is determined
by a matrix $B$ such that
\begin{equation}
  ih_-(\varepsilon_1,\varepsilon_2) = \varepsilon_1^\dagger B \varepsilon_2~,
\end{equation}
and the defining property~\eqref{eq:spinor-ips} becomes
\begin{equation}
  \Gamma_a^\dagger B = - B \Gamma_a~,
\end{equation}
which says that $B$ anticommutes with $\Gamma_0$ and commutes with the
rest.  In other words, $B$ must be proportional to $\Gamma_{12345}$,
which in this realisation is given by
\begin{equation}
  \Gamma_{12345} = i \sigma_2 \otimes \sigma_3 \otimes \sigma_2~,
\end{equation}
which is symmetric and imaginary, hence skewhermitian, as expected.
We define $B := - \Gamma_{12345}$, where the sign is for later
convenience.

The symmetric inner product $g_-$ is given by a matrix $C$ such that
\begin{equation}
  g_-(\varepsilon_1,\varepsilon_2) = \varepsilon_1^T C \varepsilon_2~,
\end{equation}
where now
\begin{equation}
  \Gamma_a^T C = - C \Gamma_a~,
\end{equation}
which says that $C$ commutes with the skewsymmetric $\Gamma_a$, namely
$\Gamma_{2,5}$, and anticommutes with $\Gamma_{0,1,3,4}$.  This means
that $C$ must be proportional to $\Gamma_{0134}$, which in this
realisation is given by
\begin{equation}
  \Gamma_{0134} = i \sigma_2 \otimes \sigma_1 \otimes \sigma_2~,
\end{equation}
which is imaginary and symmetric.  We will define $C :=
\Gamma_{0134}$.

For the $(2,0)$-theory we will also need an explicit realisation of
$\Cl(0,5)$, conveniently given by the following $4\times 4$ matrices
\begin{equation}
  \gamma_1 = \sigma_1 \otimes \sigma_2 \quad 
  \gamma_2 = \sigma_2 \otimes \sigma_2 \quad
  \gamma_3 = -\sigma_3 \otimes \sigma_2 \quad
  \gamma_4 = \1 \otimes \sigma_3 \quad
  \gamma_5 = \1 \otimes \sigma_1~,
\end{equation}
with $\1 = \1_2$ again.

The invariant quaternionic structure $j$ is given by the composition
$m_j \circ \chi$, with $\chi$ again complex conjugation and $m_j$ a
matrix satisfying $m_j \overline{m_j} = -\1_4$ and $m_j \gamma_i =
\overline{\gamma_i} m_j$.  We can therefore take $m_j =
\gamma_{245}$.

The $\HH$-hermitian invariant inner product $\left<-,-\right>_+$
decomposes into $h_+ + j \omega_+$, where $h_+$ is $\CC$-hermitian and
$\omega_+$ is $\CC$-symplectic.  In this realisation, $h_+$ is
defined in terms of a matrix $b$ by
\begin{equation}
  h_+(\varepsilon_1,\varepsilon_2) = \varepsilon_1^\dagger b
  \varepsilon_2~,
\end{equation}
where $\gamma_i^\dagger b = b \gamma_i$.  Since all $\gamma_i^+ =
\gamma_i$ for all $i$, we can choose $b = \1_4$ without loss of
generality.  The $\CC$-symplectic inner product $\omega_+$ is given in
terms of a matrix $c$ by
\begin{equation}
  \omega_+(\varepsilon_1,\varepsilon_2) = \varepsilon_1^T c
  \varepsilon_2~,
\end{equation}
where $\gamma_i^T c = c \gamma_i$.  Thus $c$ must commute with
$\gamma_{2,4,5}$ and anticommute with $\gamma_{1,3}$, whence we can
take $c = \gamma_{245}$ which is real and symplectic.

In the tensor product representation $S_+ \otimes S_2$, the
conjugation $\rC = J \otimes j$ is given explicitly by
\begin{equation}
  \rC = \left(\Gamma_{025} \otimes \gamma_{245}\right) \circ \chi~,
\end{equation}
so that an element $\varepsilon$ of $\rS$ obeys
\begin{equation}
  \overline{\varepsilon} = \left(\Gamma_{025} \otimes
    \gamma_{245}\right) \varepsilon~.
\end{equation}

Let $\varepsilon_{1,2} \in \rS$.  Then it is an easy calculation to
show that the sesquilinear and bilinear inner products agree, as
expected.  This is nothing but the fact that for a Majorana spinor,
the Dirac conjugate agrees with the Majorana conjugate; explicitly,
\begin{equation}
  \varepsilon_1^\dagger (B \otimes b) \varepsilon_2 = \varepsilon_1^T
  (\Gamma_{025} \otimes \gamma_{245})^T (B\otimes b) \varepsilon_2 = \varepsilon_1^T
  (C\otimes c) \varepsilon_2~.
\end{equation}

\subsection{Vectors, forms and their Clifford action}
\label{sec:vectors-forms-their}

As a vector space, the Clifford algebra is isomorphic (as a
$\ZZ_2$-graded vector space) to the exterior
algebra.  When we globalise, the Clifford bundle $\Cl(TM)$ is
isomorphic as a $\ZZ_2$-graded vector bundle to the bundle of differential forms
$\Omega^*(M)$.  This means that differential forms can act on
spinors.  If $\theta \in \Omega^k(M)$ is a differential form of rank
$k$ and $\varepsilon \in C^\infty(M;\eS)$ is a spinor field, then we
will denote $\theta \cdot \varepsilon$ the spinor field obtained by
Clifford acting with $\theta$ on $\varepsilon$.  Explicitly,
\begin{equation}
  \theta \cdot \varepsilon = \tfrac1{k!} \theta_{a_1\dots a_k}
  \Gamma^{a_1\dots a_k} \varepsilon~.
\end{equation}
Similarly, if $X \in C^\infty(M;TM)$ is a vector field, we can define
its Clifford action $X \cdot \varepsilon$ on a spinor field as the
Clifford action of its dual 1-form $X^\flat$.

Let $\nu \in \Omega^6(M)$ denote the volume form.  Its Clifford action
is via $\Gamma_7 = \Gamma^{012345}$.  Then if $\theta \in
\Omega^k(M)$ and $\varepsilon$ is any spinor field,
\begin{equation}
  \Gamma_7 \theta \cdot \varepsilon = -(\star \theta) \cdot \varepsilon~,
\end{equation}
where $\star \theta$ is the Hodge dual.  A very useful consequence of
this calculation is the following.

\begin{lemma}
  \label{lem:asdkillspos}
  Let $H \in \Omega^3_-(M)$ be an anti-selfdual 3-form and let
  $\varepsilon \in C^\infty(M;\eS_+)$ be a positive-chirality spinor
  field.  Then $H \cdot \varepsilon = 0$.
\end{lemma}

\begin{proof}
  Let $\Gamma_7$ denote the volume element in the Clifford algebra, so
  that the volume form acts via $\Gamma_7$.  Since $\varepsilon$ has
  positive chirality, $\Gamma_7 \varepsilon = \varepsilon$.  If $H \in
  \Omega^3(M)$ is any $3$-form, then $H \Gamma_7 = - \Gamma_7 H$,
  whence on the one hand
  \begin{equation*}
    \Gamma_7 H \cdot \varepsilon = - H \cdot \Gamma_7 \varepsilon = - H\cdot \varepsilon
  \end{equation*}
  and on the other hand, for $H$ anti-selfdual
  \begin{equation*}
   \Gamma_7 H \cdot \varepsilon = - (\star H) \cdot \varepsilon =  H
   \cdot \varepsilon~.
  \end{equation*}
\end{proof}

Also useful are the following identities, where $\theta \in
\Omega^k(M)$ and $X$ is any vector field:
\begin{equation}
  \label{eq:cliffordintext}
  \begin{aligned}
    X^\flat \cdot \theta - (-1)^k \theta \cdot X^\flat &= - 2 \iota_X \theta\\
    X^\flat \cdot \theta + (-1)^k \theta \cdot X^\flat &= + 2 X^\flat \wedge \theta ~.
  \end{aligned}
\end{equation}

Two more useful consequences of the Clifford relations are
\begin{equation}
  \label{eq:cliffrels}
  \Gamma^a \Gamma_b \Gamma_a = 4 \Gamma_b \qquad\text{and}\qquad
  \Gamma^a \Gamma_{bcd} \Gamma_a = 0~.
\end{equation}

\subsection{Fierz formulae}
\label{sec:fierz-formulae}

In this section we derive two important Fierz formulae.

\subsubsection{The (1,0) Fierz formula}
\label{sec:1-0-fierz}

Let us first of all consider the $(1,0)$ theory.  Let $\varepsilon \in
\eS_+$.  By choosing a complex basis $e_A$, $A=1,2$, for the
fundamental two-dimensional representation $S_1$ of $\USp(2)$ relative
to which the invariant complex symplectic form is given by the
Levi-Civita symbol $\epsilon_{AB}$, we may decompose $\varepsilon \in
\eS_+$ as $\varepsilon = \varepsilon^A e_A$, where each $\varepsilon^A
\in S_+$ is a chiral spinor of $\Spin(5,1)$.  In addition, the
$\varepsilon^A$ satisfy a reality condition whose explicit form we
will not need.  The real symplectic inner product on $\eS_+ \oplus
\eS_-$ is such that if $\varepsilon, \eta \in \eS_+\oplus \eS_-$, then
\begin{equation}
  \left<\varepsilon, \eta\right> = \epsilon_{AB}
  \left(\varepsilon^A,\eta^B\right)~,
\end{equation}
where $\left(-,-\right)$ is the symmetric inner product on $S_+ \oplus
S_-$, which we had denoted $g_-$ in
Section~\ref{sec:clifford-module-its}.

Now let $\psi_{1,2} \in S_+$ and consider the complex linear map
$\psi_1 \otimes \psi_2^\flat: S_- \to S_+$ defined by
\begin{equation}
  \left(\psi_1 \otimes \psi_2^\flat\right)(\psi_3) =
  \left(\psi_2,\psi_3\right) \psi_1~.
\end{equation}
This can be extended to an endomorphism of $P = S_+ \oplus S_-$ by
declaring it to be zero on $S_+$ and hence it defines an element of
the Clifford algebra $\Cl(5,1)$, which is the endomorphism algebra of
$P$.  Since the map reverses chirality, it lives in
$\Cl(5,1)^{\text{odd}}$, whence it is a linear combination of products
of an odd number of $\Gamma_a$ and since it annihilates $S_+$, it
takes the form
\begin{equation}
  \psi_1 \otimes \psi_2^\flat = \left(c^a \Gamma_a + \tfrac16 c^{abc}
    \Gamma_{abc}  \right) \Pi_-~,
\end{equation}
for some $c^a$ and $c^{abc}$ to be determined and where $\Pi_- =
\tfrac12 \left(\1 - \Gamma_7\right)$ is the projector onto
negative chirality spinors.

It is a simple matters of taking the trace of $\left( \psi_1 \otimes
  \psi_2^\flat \right)\Gamma_b$ and $\left( \psi_1 \otimes
  \psi_2^\flat \right)\Gamma_{abc}$ to determine that
\begin{equation}
  c^a = \tfrac14 \left(\psi_1, \Gamma^a \psi_2\right)
  \qquad\text{and}\qquad
  c^{abc} = \tfrac14 \left(\psi_1, \Gamma^{abc} \psi_2\right)~,
\end{equation}
whence we arrive at the Fierz identity
\begin{equation}
  \label{eq:Fierz10}
    \psi_1 \otimes \psi_2^\flat = \tfrac14
    \left(\psi_1,\Gamma^a\psi_2\right) \Gamma_a \Pi_- + \tfrac1{24}
    \left(\psi_1, \Gamma^{abc} \psi_2\right) \Gamma_{abc} \Pi_-~.
\end{equation}

If now $\varepsilon_{1,2} \in \eS_+$ and we apply the above Fierz
formula to the linear map $\varepsilon_1^A \otimes
(\varepsilon_2^B)^\flat: S_- \to S_+$, we arrive at
\begin{equation}
  \label{eq:Fierz10ABC}
    \varepsilon_1^A \otimes \left(\varepsilon_2^B\right)^\flat = \tfrac14
    \left(\varepsilon_1^A,\Gamma^a\varepsilon_2^B\right) \Gamma_a \Pi_- + \tfrac1{24}
    \left(\varepsilon_1^A, \Gamma^{abc} \varepsilon_2^B\right) \Gamma_{abc} \Pi_-~.
\end{equation}

A simple consequence of this Fierz identity is the following result.

\begin{lemma}
  \label{lem:3form10}
  Let $\varepsilon \in \eS_+$.  Then for all $A,B,C=1,2$,
  \begin{equation*}
   \left(\varepsilon^A, \Gamma^a \varepsilon^B\right) \Gamma_a
    \varepsilon^C = 0~.
  \end{equation*}
\end{lemma}

\begin{proof}
  An immediate consequence of the Fierz identity~\eqref{eq:Fierz10ABC}
  and equation~\eqref{eq:cliffrels} is that
  \begin{equation*}
    X^{ABC} := \left(\varepsilon^A, \Gamma^a \varepsilon^B\right) \Gamma_a
    \varepsilon^C
  \end{equation*}
  is invariant under cyclic permutations of its indices: $X^{ABC} =
  X^{BCA} = X^{CAB}$.  It also follows from the fact that $\Gamma^a$
  is skewsymmetric relative to the symmetric inner product
  $\left(-,-\right)$, that
  \begin{equation*}
    X^{ABC} = - X^{BAC}~.
  \end{equation*}
  In other words, $X^{ABC}$ is totally skewsymmetric, but since
  $A,B,C = 1,2$, it has to vanish.
\end{proof}

\subsubsection{The (2,0) Fierz formula}
\label{sec:2-0-fierz}

Every $\varepsilon \in \rS_+$ defines a linear map
$\varepsilon\otimes\varepsilon^\flat : \rS_- \to \rS_+$ by
\begin{equation}
  \left(\varepsilon\otimes\varepsilon^\flat\right)(\varepsilon') =
  \left<\varepsilon,\varepsilon'\right>\varepsilon~,
\end{equation}
with $\left<-,-\right>$ the symplectic inner product on $\rS = \rS_+
\oplus \rS_-$.  The linear map $\varepsilon\otimes\varepsilon^\flat$
extends to an endomorphism of $\rS$ which is trivial on $\rS_+$ and
hence can be expressed as an element of $\Cl(5,1) \otimes \Cl(0,5)$.
Symmetry and chirality imply that
\begin{equation}
  \varepsilon\otimes\varepsilon^\flat = c^a \Gamma_a \Pi_- + c^{a\,i}
  \Gamma_a \gamma_i \Pi_- + \tfrac1{12} c^{abc\,ij} \Gamma_{abc}\gamma_{ij}\Pi_-~,
\end{equation}
for some coefficients $c^a$, $c^{a\,i}$ and $c^{abc\,ij}$ which must
be determined.  Taking traces and remembering that $\gamma_i$ are
$4\times 4$ matrices, we find that
\begin{equation}
  \label{eq:20Fierz}
  \varepsilon\otimes\varepsilon^\flat = -\tfrac1{16}
  \left(\left<\varepsilon,\Gamma^a\varepsilon\right> \Gamma_a +
    \left<\varepsilon,\Gamma^a\gamma^i\varepsilon\right> \Gamma_a \gamma_i +
    \tfrac1{24} \left<\varepsilon, \Gamma^{abc}\gamma^{ij}\varepsilon\right>
    \Gamma_{abc}\gamma_{ij}\right)\Pi_-~.
\end{equation}

A consequence of this Fierz identity is that if $V^\mu =
\left<\varepsilon,\Gamma^\mu\varepsilon\right>$ and $\theta_\mu^i =
\left<\varepsilon,\Gamma_\mu\gamma^i\varepsilon\right>$, then
\begin{equation}
  5 V^\mu \Gamma_\mu \varepsilon + \theta_\mu^i \Gamma^\mu\gamma_i
  \varepsilon = 0~;
\end{equation}
although in contrast with the $(1,0)$ case, $V_\varepsilon$ does not
Clifford annihilate $\varepsilon$.  In particular, $V_\varepsilon$ is
not necessarily null, but only causal.

\bibliographystyle{utphys}
\bibliography{Geometry,Algebra,Sugra}

\providecommand{\href}[2]{#2}\begingroup\raggedright\begin{thebibliography}{10}

\bibitem{DNP}
M.~Duff, B.~Nilsson, and C.~Pope, ``Kaluza--{K}lein supergravity,'' {\em Phys.
  Rep.} {\bf 130} (1986) 1--142.

\bibitem{FMPHom}
J.~M. Figueroa-O'Farrill, P.~Meessen, and S.~Philip, ``Supersymmetry and
  homogeneity of {M}-theory backgrounds,'' {\em Class. Quant. Grav.} {\bf 22}
  (2005) 207--226, \href{http://arxiv.org/abs/hep-th/0409170}{{\tt
  arXiv:hep-th/0409170}}.

\bibitem{EHJGMHom}
J.~M. Figueroa-O'Farrill, E.~Hackett-Jones, and G.~Moutsopoulos, ``The
  {K}illing superalgebra of ten-dimensional supergravity backgrounds,'' {\em
  Class. Quant. Grav.} {\bf 24} (2007) 3291--3308,
  \href{http://arxiv.org/abs/hep-th/0703192}{{\tt arXiv:hep-th/0703192}}.

\bibitem{FigueroaO'Farrill:2012fp}
J.~Figueroa-O'Farrill and N.~Hustler, ``{The homogeneity theorem for
  supergravity backgrounds},'' {\em JHEP} {\bf 1210} (2012) 014,
\href{http://arxiv.org/abs/1208.0553}{{\tt arXiv:1208.0553 [hep-th]}}.

\bibitem{Gutowski:2003rg}
J.~B. Gutowski, D.~Martelli, and H.~S. Reall, ``{All Supersymmetric solutions
  of minimal supergravity in six- dimensions},'' {\em Class.Quant.Grav.} {\bf
  20} (2003) 5049--5078,
\href{http://arxiv.org/abs/hep-th/0306235}{{\tt arXiv:hep-th/0306235
  [hep-th]}}.

\bibitem{CFOSchiral}
A.~Chamseddine, J.~M. Figueroa-O'Farrill, and W.~A. Sabra, ``Vacuum solutions
  of six-dimensional supergravities and lorentzian {L}ie groups,''
  \href{http://arxiv.org/abs/hep-th/0306278}{{\tt arXiv:hep-th/0306278}}.

\bibitem{NishinoSezgin10}
H.~Nishino and E.~Sezgin, ``Matter and gauge couplings of ${N}{=}2$
  supergravity in six dimensions,'' {\em Phys. Lett.} {\bf B144} (1984) 187.

\bibitem{JMFKilling}
J.~M. Figueroa-O'Farrill, ``On the supersymmetries of {A}nti-de~{S}itter
  vacua,'' {\em Class. Quant. Grav.} {\bf 16} (1999) 2043--2055,
  \href{http://arxiv.org/abs/hep-th/9902066}{{\tt arXiv:hep-th/9902066}}.

\bibitem{Kosmann}
Y.~Kosmann, ``Dérivées de {L}ie des spineurs,'' {\em Annali di Mat. Pura
  Appl. (IV)} {\bf 91} (1972) 317--395.

\bibitem{Townsend20}
P.~K. Townsend, ``A new anomaly free chiral supergravity theory from
  compactification on ${K}3$,'' {\em Phys. Lett.} {\bf B139} (1984) 283.

\bibitem{Tanii20}
Y.~Tanii, ``${N}{=}8$ supergravity in six dimensions,'' {\em Phys. Lett.} {\bf
  B145} (1984) 197--200.

\bibitem{Harvey}
F.~R. Harvey, {\em Spinors and calibrations}.
\newblock Academic Press, 1990.

\end{thebibliography}\endgroup

\end{document}